\documentclass[a4paper, 10pt, conference]{ieeeconf}      % Use this line for a4 paper

\IEEEoverridecommandlockouts                              % This command is only needed if 
\usepackage{amsmath} % assumes amsmath package installed
\usepackage{amssymb}  % assumes amsmath package installed
\usepackage{cite}
\usepackage{amsthm}
\usepackage{lipsum}
\usepackage{graphicx}
\newtheorem{theorem}{Theorem}[section]
\newtheorem{proposition}[theorem]{Proposition}
\title{\LARGE \bf
When would online platforms pay data dividends?
}

\author{Sukanya Kudva and Anil Aswani%<-this % stops a space
\thanks{*This material is based upon work supported by the National Science Foundation under Grant CMMI-1847666.}% <-this % stops a space
\thanks{S Kudva and A Aswani are with the Department of Industrial Engineering
and Operations Research, University of California, Berkeley, CA 94720, USA
        {\tt\small sukanya\_kudva@berkeley.edu, aaswani@berkeley.edu}}%
}

\begin{document}
\maketitle
\thispagestyle{empty}
\pagestyle{empty}

%%%%%%%%%%%%%%%%%%%%%%%%%%%%%%%%%%%%%%%%%%%%%%%%%%%%%%%%%%%%%%%%%%%%%%%%%%%%%%%%
\begin{abstract}
Online platforms, including social media and search platforms, have routinely used their users' data for targeted ads, to improve their services, and to sell to third-party buyers. But an increasing awareness of the importance of users' data privacy has led to new laws that regulate data-sharing by platforms. Further, there have been political discussions on introducing data dividends, that is paying users for their data. Three interesting questions are then: When would these online platforms be incentivized to pay data dividends? How does their decision depend on whether users value their privacy more than the platform's free services? And should platforms invest in protecting users' data? This paper considers various factors affecting the users' and platform's decisions through utility functions. We construct a principal-agent model using a Stackelberg game to calculate their optimal decisions and qualitatively discuss the implications. Our results could inform a policymaker trying to understand the consequences of mandating data dividends. 
\end{abstract}

%%%%%%%%%%%%%%%%%%%%%%%%%%%%%%%%%%%%%%%%%%%%%%%%%%%%%%%%%%%%%%%%%%%%%%%%%%%%%%%%
\section{INTRODUCTION}
The revenue models of many online platforms depend on collecting, analyzing, and selling users' data. A free-service and advertising-based revenue model can cause conflicts in the users' and platform's interests. Users may be concerned about their privacy and possible misuse of data, while platforms want to maximize their profits. Further, users' perception of a platform's ethics and their willingness to participate can be affected by the platform's revenue model, pricing decisions, and privacy practices\cite{1,2,3}. 
\subsection{Cybersecurity on online platforms}
After the onset of the Internet of Things (IoT), there has been an increase in the variety, speed, and volume of users' data collected. Information from multiple sources including devices, sensors, and social networks is being used by platforms to assist users and collect data \cite{3'}. 

Though users have become more aware of the rampant collection of their data, they often do not know about the proliferation of IoT devices in their everyday lives. For instance, in August 2022 the Australian federal court convicted a major search platform for collecting users' location data without their knowledge \cite{5'}. When users give consent and permissions to apps and platforms, they tend to underestimate the implications of it \cite{6'}. Reforms to protect users' privacy are very much needed across the world \cite{ja}. Users should be able to control, delete and transfer their data across different platforms and service providers. They should be asked for explicit consent every time a platform wants to use their data for a new purpose \cite{3'}. 
\subsection{Privacy and data dividends}
With growing user concerns, consumer privacy legislation has become an important topic for public discussion, and multiple new data privacy laws have been introduced \cite{4}. In Europe, the General Data Protection Regulation (GDPR) was introduced in 2016 to give consumers more control over their data \cite{7'}. In 2019, legislators in California announced their intent to introduce data dividends, which is a model in which platforms would pay users in exchange for use of users' data \cite{5}. The same year, Oregon legislators introduced a bill, called the Health Information Property Act, to compensate consumers for monetizing their health data \cite{8'}. 

Implementing data dividends comes with its own challenges as companies holding the data are far more powerful than individual users. Further, there is a huge information asymmetry and only companies know the actual value of the users' data. The critics of data dividends argue that selling data would make it a commodity and be counter-productive in protecting users' privacy. They also feel that vulnerable groups -- such as people of color and the poor -- who are currently discriminated against should not be incentivized to pour more data into the system and further reduce their privacy \cite{9',10'}. On the other hand, the proponents of data dividends argue that today's technology economy is hugely driven by monetizing users' data, and paying users a share of these benefits is only fair.  

Recent studies have explored different ways of pricing data dividends for each user based on the value of their individual data \cite{6,7}. Using the idea of Shapley and Owen values, they calculated the contribution made by each user to the platform's profits. Some scholars have also proposed different ways of charging data dividends and how they could be used for the greater public good \cite{11'}. For our work, we ask different questions: Should platforms pay data dividends at all, and why? In this paper, we do our analysis with homogeneous users but our methods can be extended to analyze heterogeneous users too.

\subsection{Contributions and outline}

Our paper is organized as follows: Sect. \ref{sect2} outlines our utility functions for an online platform and its users and our principal-agent model, Sect. \ref{sect3} analytically solves the principal-agent model in order to derive their optimal choices and Sect. \ref{sect4} discusses insights from our model.

Our paper comes in the context of rising debates on data dividends. We try to understand when and how much online platforms would pay as data dividends. We consider users' privacy concerns as an important factor, for which a platform can invest in data protection. The platform can also pay users to incentivize taking risks and sharing their data. Our main contributions are to introduce a principal-agent model using a Stackelberg game to capture the platform-users dynamics and then use this to gain a better understanding of when the platform would pay users with data dividends.

\section{OUR MODEL}\label{sect2}
An online platform provides its users with free services and a data dividend in exchange for their data. The users are allowed to choose between two levels -- high or low -- of data sharing. For each level, the platform provides a different data dividend and set of services. As the platform collects users' data, it faces the risk of a possible data breach. If a data breach occurs, the platform loses its reputation and faces possible legal and financial complications, and the users are harmed by the misuse of their data. Hence, the platforms consider investing in protecting their users' data. This could include the costs of building better technological infrastructure and signing insurance contracts.

\subsection{Platform's utility}
The platform has a fixed cost of $S$ to maintain and provide its services. It invests $I$ in users' data protection, reducing the probability of a data breach to $B(I)$. Here, $B$ is assumed to be a twice differentiable function such that $B(I) > 0$ with $\lim_{I \rightarrow \infty} B(I) = 0$, $B'(I)< 0$ and $B''(I)> 0 \; \forall I$. In case of a data breach, the platform has a loss of $F$ from legal cases and a lost reputation. When having access to $k$ users' data, the platform makes a revenue of $U(k,b)$, where $b$ of the total $k$ users chose the low level of data sharing. The revenue may be due to selling the data to a third-party, selling advertisements to display to users or other sources. The platform pays a share of this revenue to users as data dividends, which are priced at two scales -- $p_0$ and $p_1$ -- for the low and high levels of data sharing respectively.

Considering these costs, revenue, and risks, the platform has a total expected utility of :
\begin{align} \label{eq1}
  U(k,b) \-- B(I)F \-- I \-- S \-- p_0 b \-- p_1(k\--b).   
\end{align}

\subsection{User's utility}
We consider $k$ homogeneous users who have the same behavior and parameters. Each user values their personal data at $V$ and faces an additional personal loss of $L$ on a data breach. They also benefit from the platform's free services, which amount to a value of $W$. When a user chooses the low level of data sharing, these benefits and losses scale down by a factor $\alpha \in (0,1)$. The platform pays users with data dividends at two scales of pay and thereby encourages a particular level of data sharing. 

Let $c_i$ be user $i$'s decision variable so that $c_i = 0$ and $1$ for the low and high levels of data sharing respectively. Then a user $i$'s total expected utility is:
\begin{align} \label{eq2}
    c_i (p_1 \-- \mathcal{V}) + (1 \-- c_i)( p_0 - \alpha \mathcal{V}),
\end{align}
where $\mathcal{V} = \Bar{V} + B(I) L$ and $\Bar{V} =V \-- W $.

\subsection{Principal-agent formulation}
We construct a Stackelberg game \cite{12',13',14'}, in which players act sequentially with follower(s) acting after a leader. In our model, the platform first prices data dividends for different levels of data-sharing, and decides on investment for users' data protection. Given this information, the users decide how much data to share. 

The optimal, equilibrium choices of the platform and the users are their Stackelberg strategies. We capture this using an optimization problem, which maximizes the platform's utility when the users are maximizing their utilities\cite{8}. It is formulated as follows:
\begin{align} \label{eq3}
\begin{split}
    \max &{\;\; U(k,b) \-- B(I)F \-- I \-- S \-- p_0 b \-- p_1(k\--b)} \\
    s.t &\; \;c_i^* = \text{arg}\max_{\;c_i \in \{0,1\}} c_i( p_1 \--\mathcal{V}) + (1\--c_i)(p_0 \-- \alpha \mathcal{V}) \\
    & \hspace{150pt} \forall i =\{1, \cdots, k\}\\
    & \;\;c_i^*( p_1 \--\mathcal{V}) + (1\--c_i^*)(p_0 \-- \alpha \mathcal{V}) \geq 0 \\
    & \hspace{150pt} \forall i=\{1, \cdots, k\}\\
    & \;\;I,\: p_0 \:, p_1 \geq 0 . 
\end{split} 
\end{align} 
The second constraint in (\ref{eq3}) ensures that the users do not have a net loss from using the platform, without which they would leave the platform.
\section{Optimal choices } \label{sect3}
Since the users are homogeneous, they make similar choices: either $c_i^* = 1$ or $c_i^* = 0$ for every user $i$. If users find both levels of data sharing to be utility-maximizing, then we assume that they all choose the level that most benefits the platform. We solve these two cases of $c_i^*$ separately. The optimal solution is the best of the two cases and can vary with the numerical values of the parameters of the model.
\subsection{Case 1: $c_i^* = 1 \;\; \forall i = \{1, \cdots , k\}$}
To ensure users choose high level of data sharing, the platform must price the data dividends such that $p_1 \-- \mathcal{V} \geq p_0 - \alpha \mathcal{V}$. Note that $b= 0$ in this case. The optimization problem then reduces to: 
\begin{align} \label{eq4}
\begin{split}
    \max & {\; U(k,0) \-- B(I)F \-- I \-- p_1 k \--S} \\
    s.t \; \; & p_1 \-- \mathcal{V} \geq p_0 \-- \alpha \mathcal{V}\\
    & p_1 \-- \mathcal{V} \geq 0 \\
    & I,\: p_0 \:, p_1 \geq 0.
\end{split}
\end{align}
Since $p_0 \geq 0$, $p_0^* = 0$ is optimal for the problem. Given this, one can conclude $p_1^* = \mathcal{V}$ when $\mathcal{V} \geq 0$ and $p_1^* = 0$ when $\mathcal{V} \leq 0$. Also, $U(k,0) -S$ can be treated as a constant. These observations further reduce (\ref{eq4}) to two sub-cases with optimization problems in a single variable $I$. 
%\vspace{2.5mm}
\subsubsection{Sub-case 1 - If $\mathcal{V} \geq 0$ then $p_1^* = \mathcal{V}$} \label{type1}

Set $p_1^* = \mathcal{V} = \Bar{V} + B(I) L$ and add an additional constraint $\mathcal{V} \geq 0$ in (\ref{eq4}).
\begin{align} \label{eq6}
\begin{split}
    \min & {\; B(I) F_1 + I + \Bar{V} k} \\
    s.t \; \; & \Bar{V} + B(I) L \geq 0\\
    & I \geq 0,
\end{split}
\end{align}
where $F_1 = F + Lk$.

Let $I_1 = B'^{-1}(\frac{\--1}{F_1})$ and $I_2 = B^{-1} (\frac{\--\Bar{V}}{L})$.
\begin{proposition} \label{prop1}
    When $\Bar{V} \geq 0$, $I^* = I_1$ if $I_1$ exists. Else, $I^* =0$. 
\end{proposition}
%\renewcommandsymbol{$\blacksquare$}
\begin{proof}
   $B(I)L \geq \--\Bar{V}$ holds $\forall I$ because $\Bar{V} \geq 0$. When $I_1$ exists, it's the optimal solution to the unconstrained problem. Since it also satisfies $I_1 \geq 0$, it's the optimal solution for the constrained problem too. When $I_1$ doesn't exist, the objective in (\ref{eq6}) is an increasing function for $I \geq 0$. So the smallest feasible value of $I$, i.e $0$, is optimal. 
\end{proof}

\begin{proposition}\label{prop2}
    When $\Bar{V} \leq 0$,
    \begin{enumerate}
        \item If $I_1$ exists, $I^* = I_1$ when $B(I_1)L \geq \--\Bar{V}$. Else when $B(I_1)L < \--\Bar{V}$, $I^* = I_2$ if $I_2$ exists and no solution otherwise.
        \item If $I_1$ doesn't exist, $I^* =0 $ when $B(0) L \geq \-- \Bar{V}$ and no solution otherwise.  
    \end{enumerate}
\end{proposition}
%\renewcommandsymbol{$\blacksquare$}
\begin{proof}
    When $I_1$ exists and $B(I_1)L \geq \--\Bar{V}$ holds, $I_1$, which is the unconstrained optimal solution, satisfies both constraints and is optimal. When $B(I_1)L < \--\Bar{V}$, $I \geq I_1$ doesn't satisfy the first constraint in (\ref{eq6}) and the objective is a decreasing function for $I < I_1$. If $I_2$ exists, it is the largest $I < I_1$ satisfying $\Bar{V} + B(I_2) L \geq 0$ and is optimal. If $I_2$, which satisfies $\Bar{V} + B(I_2)L =0$, doesn't exist then there is no solution as $B(I)$ is a strictly decreasing function and the constraint $ B(I) L \geq -\Bar{V}$ can never be satisfied.

    When $I_1$ doesn't exist, the objective is an increasing function for $I \geq 0$. If $B(0) L \geq \-- \Bar{V}$ then $I = 0 $ satisfies both constraints of the problem and is optimal. If $B(0)L < \-- \Bar{V}$ then there is no solution as $B(I)$ is strictly decreasing and the constraint $B(I) L \geq \--\Bar{V}$ can never hold. 
\end{proof}
%\vspace{2.5mm}
\subsubsection{Sub-case 2 - If $\mathcal{V} \leq 0$ then $p_1^* = 0$} \label{type2}
Set $p_1 = 0$ and add an additional constraint $\mathcal{V} \leq 0$ in (\ref{eq4}).
\begin{align} \label{eq7}
\begin{split}
    \min & {\;B(I) F_0 + I } \\
    s.t \; \; & \Bar{V} + B(I) L \leq 0\\
    & I \geq 0
\end{split}
\end{align}
where $F_0 = F$.

Let $I_3 = B'^{-1}(\frac{\--1}{F_0})$ and as before, $I_2 = B^{-1} (\frac{\--\Bar{V}}{L})$.
\begin{proposition}\label{prop3}
    When $\Bar{V} > 0$, there is no solution.
\end{proposition}
\begin{proof}
    $\--\Bar{V} \geq B(I) L$ can never hold for any $I$ as $\Bar{V} > 0$.
\end{proof}

\begin{proposition}\label{prop4}
    When $\Bar{V} \leq 0$, 
    \begin{enumerate}
        \item If $I_3$ exists, $I^* = I_3$ when $\-- \Bar{V} \geq B(I_3) L$. Else when $\-- \Bar{V} < B(I_3) L$, $I^* =I_2$.
        \item If $I_3$ doesn't exist, $I^* =0 $ when $\--\Bar{V} \geq B(0) L$ and $I^*=I_2$ otherwise.
    \end{enumerate}
\end{proposition}
\begin{proof}
    When $I_3$ exists and $\-- \Bar{V} \geq B(I_3) L$ holds, $I_3$, which is the unconstrained optimal solution, satisfies both constraints and is optimal. When $\-- \Bar{V} < B(I_3) L$, $I \leq I_3$ doesn't satisfy $\-- \Bar{V} \geq B(I) L$ as $B(I)$ is strictly decreasing. Since the objective is increasing for $I > I_3$, $I_2$ which is the smallest $I$ satisfying $\--\Bar{V} \geq B(I)L$ is optimal. Here, $I_2$ exists because $I_2 > I_3$ and $B(I)$ is strictly decreasing.

    When $I_3$ doesn't exist, the objective is increasing for $I \geq 0$ and hence the smallest feasible $I$ is optimal. If $\--\Bar{V} \geq B(0) L$ then $I = 0$ satisfies both the constraints and is optimal. If $\--\Bar{V} < B(0) \: L$ then $I = I_2$ is the smallest feasible $I$ with $\Bar{V} + B(I)L = 0$ and is optimal. 
\end{proof}
%\vspace{2.5mm}
The results from the two sub-cases can be combined by finding the minimum objective value of the two sub-cases.

\subsubsection{Final results of Case 1} 
\begin{theorem}\label{thm1}
    For the case when all users choose the high level of data-sharing, when $\Bar{V} \leq 0$ the optimal choices of the platform are as follows:
    \begin{enumerate}
        \item If $I_1$ and $I_3$ exist, $(I^*= I_1,\: p_1^* = \Bar{V} + B(I_1)L,\: p_0^*= 0)$ when $L \geq \frac{\--\Bar{V}}{B(I_1)}$, $(I^*= I_2,\: p_1^* = 0,\: p_0^*= 0)$ when $\frac{\--\Bar{V}}{B(I_1)}\geq L \geq  \frac{\--\Bar{V}}{B(I_3)}$ and $(I^*= I_3,\: p_1^* = 0,\: p_0^*= 0)$ when $ \frac{\--\Bar{V}}{B(I_3)} \geq L$.
        \item If $I_1$ exists but not $I_3$, $(I^*= I_1,\: p_1^* = \Bar{V} + B(I_1)L,\: p_0^*= 0)$ when $L \geq \frac{\--\Bar{V}}{B(I_1)}$, $(I^*= I_2,\: p_1^* = 0,\: p_0^*= 0)$ when $\frac{\--\Bar{V}}{B(I_1)} \geq L \geq \frac{\--\Bar{V}}{B(0)}$ and $(I^*= 0,\: p_1^* = 0,\: p_0^*= 0)$ when $\frac{\--\Bar{V}}{B(0)} \geq L$.
        \item If both $I_1$ and $I_3$ don't exist, $(I^* = 0, \: p_1^* = \Bar{V} + B(0)L, \:p_0^* =0)$ when $L \geq \frac{\--\Bar{V}}{B(0)}$ and $(I^*= 0,\: p_1^* = 0,\: p_0^*= 0)$ when $\frac{\--\Bar{V}}{B(0)} \geq L$.
    \end{enumerate}
\end{theorem}
\begin{proof}
    The results from Propositions \ref{prop1}-\ref{prop4} are combined by finding the optimal of the two sub-cases.
\begin{enumerate}
    \item When $I_1$ and $I_3$ exist : 
\end{enumerate}
If $L \geq \frac{\--\Bar{V}}{B(I_1)}$, the optimal solution is one of $(I= I_1,\: p_1 = \Bar{V} + B(I_1)L,\: p_0= 0)$ and $(I= I_2,\: p_1 = 0,\: p_0= 0)$ with objective function values $B(I_1)F_1 + I_1 + \Bar{V}k$ and $B(I_2)F_0 + I_2$ respectively. The former solution is optimal if: 
        \begin{align} \label{eq8}
            \begin{split}
                B(I_2)F_0 + I_2 &\geq B(I_1)F_1 + I_1 + \Bar{V}k \\
                \iff I_2 \-- I_1 &\geq F_1 (B(I_1) \-- B(I_2)) \\
                \iff \frac{\-- 1}{F_1} &\leq \frac{B(I_1)\--B(I_2)}{I_1\--I_2}
            \end{split}
        \end{align}
        Since $B(I_1) \geq \frac{\-- \Bar{V}}{L} = B(I_2)$, $I_1 \leq I_2$ as $B$ is strictly decreasing. Using the mean value theorem, $\exists I'$ s.t $I_1 \leq I' \leq I_2$ and $B'(I') = \frac{B(I_1)\--B(I_2)}{I_1\--I_2}$. As $B'' > 0$, $B'(I_1) \geq B'(I') \implies \frac{\-- 1}{F_1} \leq \frac{B(I_1)\--B(I_2)}{I_1\--I_2}$. Hence, (\ref{eq8}) is true.
        
       If $\frac{\--\Bar{V}}{B(I_1)}\geq L \geq  \frac{\--\Bar{V}}{B(I_3)}$ then both sub-cases give the same optimal solution $(I^*= I_2,\: p_1^* = 0,\: p_0^*= 0)$.
       
      If $\frac{\--\Bar{V}}{B(I_3)} \geq L \geq \frac{\--\Bar{V}}{B(0)}$ then the optimal solution is one of $(I= I_3,\: p_1 = 0,\: p_0= 0)$ and $(I= I_2,\: p_1 = 0,\: p_0 = 0)$ with objective function values $B(I_3)F_0 + I_3$ and $B(I_2)F_0 + I_2$ respectively. The former solution is optimal if: 
        \begin{align} \label{eq9}
            \begin{split}
                B(I_2)F_0 + I_2 &\geq B(I_3)F_0 + I_3 \\
                \iff \frac{\-- 1}{F_0} &\geq \frac{B(I_3)\--B(I_2)}{I_3\--I_2}
            \end{split}
        \end{align}
        Since $B(I_3) \leq \frac{\-- \Bar{V}}{L} = B(I_2)$, $I_3 \geq I_2$ as $B$ is strictly decreasing. Using the mean value theorem, $\exists I'$ s.t $I_2 \leq I' \leq I_3$ and $B'(I') = \frac{B(I_3)\--B(I_2)}{I_3\--I_2}$. As $B'' > 0$, $B'(I_3) \geq B'(I') \implies \frac{\-- 1}{F_0} \geq \frac{B(I_3)\--B(I_2)}{I_3\--I_2}$. Hence, (\ref{eq9}) is true.
        
     If $\frac{\--\Bar{V}}{B(0)} \geq L$ then $(I^*= I_3,\: p_1^* = 0,\: p_0^*= 0)$ is only solution from the two sub-cases as one of them was infeasible.
\begin{enumerate}
    \setcounter{enumi}{1}
    \item When $I_1$ exists but not $I_3$:
\end{enumerate}
If $L \geq \frac{\--\Bar{V}}{B(I_1)}$ then $(I= I_1,\: p_1 = \Bar{V} + B(I_1)L,\: p_0= 0)$ and $(I= I_2,\: p_1 = 0,\: p_0= 0)$ are compared and (\ref{eq8}) holds. If $\frac{\--\Bar{V}}{B(I_1)} \geq L \geq \frac{\--\Bar{V}}{B(0)}$ then both sub-cases give the same optimal solution $(I^*= I_2,\: p_1^* = 0,\: p_0^*= 0)$. If $\frac{\--\Bar{V}}{B(0)} \geq L$ then $(I^*= 0,\: p_1^* = 0,\: p_0^*= 0)$ is only solution from the two sub-cases as one of them was infeasible.
\begin{enumerate}
    \setcounter{enumi}{2}
    \item When both $I_1$ and $I_3$ don't exist:
\end{enumerate}
If $L \geq \frac{\--\Bar{V}}{B(0)}$ then $(I= 0,\: p_1 = \Bar{V} + B(0)L,\: p_0= 0)$ and $(I= I_2,\: p_1 = 0,\: p_0 = 0)$ with objective function values $B(0)F_1 + \Bar{V}k$ and $B(I_2)F_0 + I_2$ respectively are compared. The former solution is optimal if: 
        \begin{align} \label{eq10}
            \begin{split}
                B(I_2)F_0 + I_2 &\geq B(0)F_1 + \Bar{V}k \\
                \iff (B(I_2) \-- B(0))F_1 &\geq \--I_2 \\
                \iff \frac{B(I_2)\--B(0)}{I_2} &\geq \frac{\-- 1}{F_1}
            \end{split}
        \end{align}
Using the mean value theorem, $\exists I'$ s.t $I_2 \leq I' \leq 0$ and $B'(I') = \frac{B(I_2)\--B(0)}{I_2}$. And $B'(I') \geq \frac{\-- 1}{F_1}$ as $I_1$ doesn't exist. Hence, (\ref{eq10}) is true.

If $\frac{\--\Bar{V}}{B(0)} \geq L$ then $(I^*= 0,\: p_1^* = 0,\: p_0^*= 0)$ is only solution from the two sub-cases as one of them was infeasible. 
\end{proof}

\begin{theorem}\label{thm2}
    For the case when all users choose the high level of data-sharing, when $\Bar{V} \geq 0$ the optimal choices of the platform are as follows:
    \begin{enumerate}
    \item If $I_1$ exists then $(I^*= I_1,\:p_1^* = \Bar{V} + B(I_1)L,\:p_0^*=0)$ is optimal.
    \item If $I_1$ doesn't exist then $(I^*= 0,\:p_1^* = \Bar{V} + B(0)L,\:p_0^*=0)$
    \end{enumerate}
\end{theorem}
\begin{proof}
     The results follow from Propositions \ref{prop1}-\ref{prop4}, and the above solutions are the only feasible solutions from the two sub-cases. 
\end{proof}

The final results of case 1 are summarized in table I.

\subsection{Case 2: $c_i^* = 0 \;\; \forall i = \{1, \cdots , k\}$}
In this case, platform must price the data dividends such that $p_0 \-- \alpha \mathcal{V} \geq p_1 \-- \mathcal{V}$. Here, $b=k$. The optimization problem becomes: 
\begin{align}\label{eq5}
\begin{split}
    \max & {\; U(k,k) \-- B(I)F \-- I \-- p_0 k \--S} \\
    s.t \; \; & p_0 \-- \alpha \mathcal{V} \geq p_1 \-- \mathcal{V}\\
    & p_0 \-- \alpha\mathcal{V} \geq 0 \\
    & I,\: p_0 \:, p_1 \geq 0. 
\end{split}
\end{align}
Similar to case 1, it can be observed that $p_1^* = 0$ is optimal and $p_0^* = \alpha \mathcal{V}$ when $\mathcal{V} \geq 0$ and $p_0^* = (\alpha\--1)\mathcal{V}$ when $\mathcal{V} \leq 0$. 

Let $I_4 = B'^{\--1}\left(\frac{\--1}{F+ \alpha Lk}\right)$, $I_5 = B'^{-1}\left(\frac{\--1}{F + (\alpha - 1)Lk}\right)$ and as before, $I_2 = B^{-1} (\frac{\--\Bar{V}}{L})$. 
\begin{theorem}
    For the case when all users choose the low level of data-sharing, when $\Bar{V} \leq 0$ the optimal choices of the platform are as follows:
    \begin{enumerate}
        \item If $I_4$ and $I_5$ exist, $(I^*= I_4,\: p_1^* = 0,\: p_0^*=\alpha(\Bar{V} + B(I_4)L))$ when $L \geq \frac{\--\Bar{V}}{B(I_4)}$, $(I^*= I_2,\: p_1^* = 0,\: p_0^*= 0)$ when $\frac{\--\Bar{V}}{B(I_4)}\geq L \geq  \frac{\--\Bar{V}}{B(I_5)}$ and $(I^*= I_5,\: p_1^* = 0,\: p_0^*= (\alpha\--1)(\Bar{V} + B(I_5)L)$ when $ \frac{\--\Bar{V}}{B(I_5)} \geq L$.
        \item If $I_4$ exists but not $I_5$, $(I^*= I_4,\: p_1^* = 0,\: p_0^*= \alpha(\Bar{V} + B(I_4)L))$ when $L \geq \frac{\--\Bar{V}}{B(I_4)}$, $(I^*= I_2,\: p_1^* = 0,\: p_0^*= 0)$ when $\frac{\--\Bar{V}}{B(I_4)} \geq L \geq \frac{\--\Bar{V}}{B(0)}$ and $(I^*= 0,\: p_1^* = 0,\: p_0^*= (\alpha\--1)(\Bar{V} + B(0)L) )$ when $\frac{\--\Bar{V}}{B(0)} \geq L$.
        \item If both $I_4$ and $I_5$ don't exist, $(I^* = 0, \: p_1^* = 0, \:p_0^* =\alpha(\Bar{V} + B(0)L))$ when $L \geq \frac{\--\Bar{V}}{B(0)}$ and $(I^*= 0,\: p_1^* = 0,\: p_0^*= (\alpha\--1)(\Bar{V} + B(0)L))$ when $\frac{\--\Bar{V}}{B(0)} \geq L$.
    \end{enumerate}
\end{theorem}
\begin{proof}
    Propositions similar to \ref{prop1}-\ref{prop4} can be derived for case 2. The proof is then similar to Theorem \ref{thm1}. 
\end{proof}

\begin{theorem}
    For the case when all users choose the low level of data-sharing, when $\Bar{V} \geq 0$ the optimal choices of the platform are as follows:
    \begin{enumerate}
    \item If $I_4$ exists then $(I^*= I_4,\:p_1^* = 0,\:p_0^*=\alpha(\Bar{V} + B(I_4)L))$ is optimal.
    \item If $I_4$ doesn't exist then $(I^*= 0,\:p_1^* = 0,\:p_0^*=\alpha(\Bar{V} + B(0)L))$
    \end{enumerate}
\end{theorem}
\begin{proof}
    The proof is similar to Theorem \ref{thm2}.
\end{proof}

The final results of case 2 are summarized in table I.

\begin{table*}[t]
%\vspace{20pt}
\caption{}
\begin{center}
\begin{tabular}{|c|c|c|}
\hline
\multicolumn{3}{|c|}{Final results of Case 1 when $\Bar{V}\leq 0$ : Here, $p_0^* = 0$ always.} \\
\hline
$I_1$ and $I_3$ exist & $I_1$ exists but not $I_3$ &  $I_1$ and $I_3$ don't exist \\
\hline
{$\!\begin{aligned}
    \text{If } L \geq \frac{\--\Bar{V}}{B(I_1)} : \; & I^* = I_1 , \\& \;p_1^* = \Bar{V} + B(I_1) L. \\
    \text{If } \frac{\--\Bar{V}}{B(I_1)} \geq L \geq \frac{\--\Bar{V}}{B(I_3)} : & \; I^* = I_2 , \\& \; p_1^* = 0. \\
    \text{If } \frac{\--\Bar{V}}{B(I_3)} \geq L : & \; I^* = I_3 , \\& \; p_1^* = 0. \\
\end{aligned}$} & {$\!\begin{aligned}
    \text{If } L \geq \frac{\--\Bar{V}}{B(I_1)} : \; & I^* = I_1 , \\& \;p_1^* = \Bar{V} + B(I_1) L. \\
    \text{If } \frac{\--\Bar{V}}{B(I_1)} \geq L \geq \frac{\--\Bar{V}}{B(0)} : & \; I^* = I_2 , \\& \; p_1^* = 0. \\
    \text{If } \frac{\--\Bar{V}}{B(0)} \geq L : & \; I^* = 0 , \\& \; p_1^* = 0. \\
\end{aligned}$}& {$\!\begin{aligned}
    \text{If } L \geq \frac{\--\Bar{V}}{B(0)} : \; & I^* = 0 , \\& \;p_1^* = \Bar{V} + B(0) L. \\
    \text{If } \frac{\--\Bar{V}}{B(0)} \geq L : & \; I^* = 0 , \\& \; p_1^* = 0. \\
\end{aligned}$} \\
\hline
\hline
\multicolumn{3}{|c|}{Final results of Case 1 when $\Bar{V} \geq 0$ : Here, $p_0^* = 0$ always.}\\
\hline
\multicolumn{2}{|c|}{$I_1$ exists} & \multicolumn{1}{|c|}{$I_1$ doesn't exist} \\
\hline
\multicolumn{2}{|c|}{$I^* = I_1$} & $I^* = 0$\\
\multicolumn{2}{|c|}{$p_1^* = \Bar{V} + B(I_1) L$} & $p_1^* = \Bar{V} + B(0) L$ \\
\hline
\hline
\multicolumn{3}{|c|}{Final results of Case 2 when $\Bar{V}\leq 0$ : Here, $p_1^* = 0$ always.}\\
\hline
$I_4$ and $I_5$ exist & $I_4$ exists but not $I_5$ &  $I_4$ and $I_5$ don't exist \\
\hline
{$\!\begin{aligned}
    \text{If } L \geq &\frac{\--\Bar{V}}{B(I_4)} : \;  I^* = I_4 , \\& \;p_0^* = \alpha(\Bar{V} + B(I_4) L). \\
    \text{If } \frac{\--\Bar{V}}{B(I_4)} \geq L \geq &\frac{\--\Bar{V}}{B(I_5)} : \; I^* = I_2 , \\& \; p_0^* = 0. \\
    \text{If } \frac{\--\Bar{V}}{B(I_5)}& \geq L : \; I^* = I_5 , \\& \; p_0^* = (\alpha \-- 1)(\Bar{V} + B(I_5) L). \\
\end{aligned}$} & {$\!\begin{aligned}
    \text{If } L \geq &\frac{\--\Bar{V}}{B(I_4)} : \; I^* = I_4 , \\& \;p_0^* = \alpha(\Bar{V} + B(I_4) L). \\
    \text{If } \frac{\--\Bar{V}}{B(I_4)} \geq L \geq &\frac{\--\Bar{V}}{B(0)} : \; I^* = I_2 , \\& \; p_0^* = 0. \\
    \text{If } \frac{\--\Bar{V}}{B(0)} &\geq L : \; I^* = 0 , \\& \; p_0^* = (\alpha \--1)(\Bar{V} + B(0) L). \\
\end{aligned}$}& {$\!\begin{aligned}
    \text{If } L \geq & \frac{\--\Bar{V}}{B(0)} : \; I^* = 0 , \\& \;p_0^* = \alpha(\Bar{V} + B(0) L). \\
    \text{If } \frac{\--\Bar{V}}{B(0)} & \geq L : \; I^* = 0 , \\& \; p_0^* = (\alpha \--1)(\Bar{V} + B(0) L). \\
\end{aligned}$} \\
\hline
\hline
\multicolumn{3}{|c|}{Final results of Case 2 when $\Bar{V} \geq 0$ : Here, $p_1^* = 0$ always.}\\
\hline
\multicolumn{2}{|c|}{$I_4$ exists} & $I_4$ doesn't exist \\
\hline
 \multicolumn{2}{|c|}{$I^* = I_4$} & $I^* = 0$ \\
\multicolumn{2}{|c|}{$p_0^* = \alpha(\Bar{V} + B(I_4) L)$} & $p_0^* = \alpha(\Bar{V} + B(0) L)$ \\
\hline
\end{tabular}
\end{center}
\end{table*}
\section{Insights from optimal choices}\label{sect4}

The platform uses data dividends to direct users to choose a particular level of data sharing. In general, the platform pays a higher data dividend for the high level of data sharing to compensate for users' larger risks. But it also earns a larger revenue when it collects more data. This trade-off between revenue and data dividend decides the optimal level of sharing.

It is optimal for the platform to pay for only one level of data sharing which maximizes its profit. However, sometimes it may prefer to not pay at all. The cases when this happens are outlined below.
\subsection{Can $p_1^* = 0$ while encouraging high level of data sharing?}
Firstly, $p_1^* \neq 0$ when $V > W$. That is, the platform always pays a data dividend whenever users value their data more than the platform's free services. When $V \leq W$, the users value the platform's free services more than their data. In this case, if the platform ensures that a data breach does not occur then it does not pay any data dividend. So if $B(I) \rightarrow 0 \; \forall I$, meaning there is a low possibility of a data breach, or $L \rightarrow 0$ so that users face a small loss from a data breach, then no data dividend is paid. Interestingly, even if $F \rightarrow \infty$, which is if the platform has heavy losses from a data breach, then they do not pay users with data dividends and instead invest heavily in data protection ($I \rightarrow \infty$). In general, the platform does not pay a data dividend when $ W \-- V - B(I^*)L \geq 0$, that is, the net users' utility before pay is positive. Currently, major technology companies do not pay data dividends. Our model suggests that these companies believe they are in the region described by the above discussion.

\subsection{Can $p_0^* = 0$ while encouraging low level of data sharing?}
Like for $p_1$, if $V > W$ then $p_0^* \neq 0$. But when $V \geq W$, $p_0^* = 0$ only if $W \-- V - B(I^*)L = 0$, that is, the net users' utility is zero. This is strange because the platform pays a data dividend when $W \-- V - B(I^*)L > 0$, that is, the net users' utility before pay is positive. The reason is that in this case the users naturally prefer a high level of data sharing. So the platform has to pay a data dividend to incentivize the low level of data sharing. One could argue that the platform should rather remove the option of a high level of data sharing, and it probably would.
\\ \\
Finally, we also analyze the trends in the investment for users' data protection. If $B(I) \rightarrow 0 \; \forall I$, meaning the possibility of a data breach is small, or if $F, L \rightarrow 0$ so that no party loses from a data breach, then the platform does not invest in data protection. If $B'(I) \rightarrow 0 \; \forall I$, that is when the likelihood of a data breach is not very sensitive to investment, then the platform chooses to not invest at all.
\begin{figure}[h!]
\includegraphics[scale=0.525]{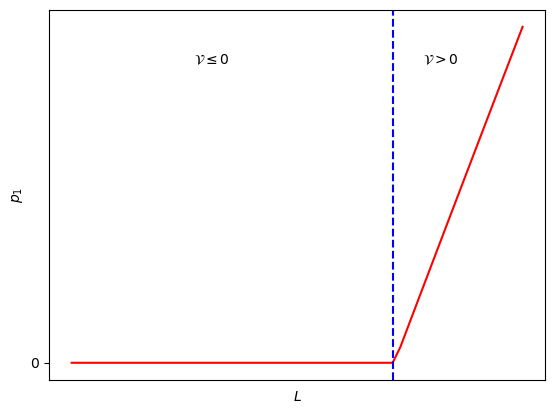}
\centering
\caption{Here, users choose the high level of data sharing, $I_1$ exists, and $\Bar{V}\leq 0$. Data dividend $p_1$ is plotted for increasing values of $L$, that is users' loss from a data breach. The platform pays a data dividend beyond a threshold when $\mathcal{V}>0$, that is users face a total loss.}
\end{figure}
\begin{figure}
    \centering
\includegraphics[scale=0.525]{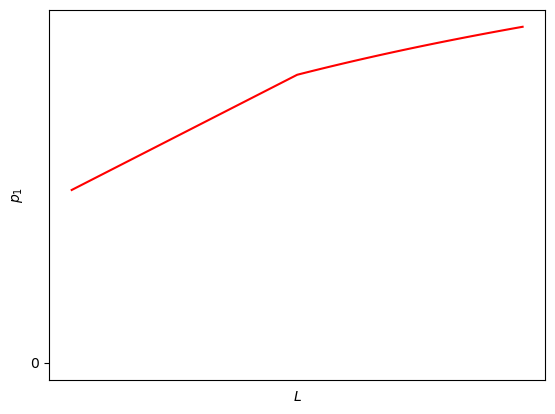}
\caption{Here, users choose the high level of data sharing and $\Bar{V}\geq 0$. Data dividend $p_1$ is plotted for increasing values of $L$, that is users' loss from a data breach. The platform always pays a data dividend of at least $\Bar{V}$.}
\end{figure}
\begin{figure}[h!]
\includegraphics[scale=0.525]{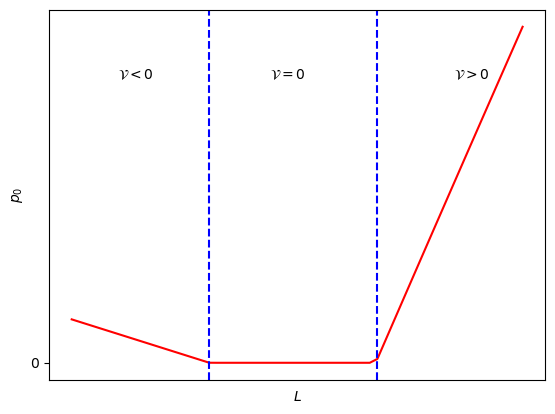}
\centering
\caption{Here, users choose the low level of data sharing, $I_4$ exists, and $\Bar{V}\leq 0$. Data dividend $p_0$ is plotted for increasing values of $L$, that is users' loss from a data breach. Interestingly, the platform doesn't pay a data dividend only when $\mathcal{V}=0$, that is users neither have a net loss nor gain.}
\end{figure}
\section{Conclusion}
We studied the dynamics between an online platform and its users when the latter are paid with data dividends. We demonstrated how an online platform could use data dividends to incentivize its users to share more data and take more risks, while also giving them an option to share as much data as they want. If the platform wants to gather more data, it pays a data dividend only when the users perceive their risks and losses to be higher than their benefits. Further, when users feel that they do not lose much from a data breach, the platform prefers to invest in data protection rather than pay users for risking a breach. On the other hand, when users perceive a high loss from a breach, the platform has to both invest and pay. It would be interesting to analyze multiple (more than 2) levels of data sharing, each with a different data dividend, and where heterogeneous users are incentivized to share more or fewer data as is optimal for the platform. We leave this for future work.

It is debatable whether data dividends could treat data as a commodity and exploit the poor and vulnerable. But if data dividends are introduced, our model shows that platforms cannot always use it to shirk responsibility for data protection. When users feel considerable harm from a data breach, platforms would both invest in data protection and pay higher data dividends. We believe that our insights could help a policymaker understand data dividends better.


%%%%%%%%%%%%%%%%%%%%%%%%%%%%%%%%%%%%%%%%%%%%%%%%%%%%%%%%%%%%%%%%%%%%%%%%%%%%%%%%

%%%%%%%%%%%%%%%%%%%%%%%%%%%%%%%%%%%%%%%%%%%%%%%%%%%%%%%%%%%%%%%%%%%%%%%%%%%%%%%%

%%%%%%%%%%%%%%%%%%%%%%%%%%%%%%%%%%%%%%%%%%%%%%%%%%%%%%%%%%%%%%%%%%%%%%%%%%%%%%%%


\begin{thebibliography}{99}
\bibitem{1}L. Xia, K. B. Monroe, and J. L. Cox, ``The Price is Unfair! A Conceptual Framework of Price Fairness Perceptions,” Journal of Marketing, vol. 68, no. 4, pp. 1–15, 2004, DOI: 10.1509/jmkg.68.4.1.42733.
\bibitem{2}Y. Su and L. Jin, ``The Impact of Online Platforms’ Revenue Model on Consumers’ Ethical Inferences,” J Bus Ethics, vol. 178, no. 2, pp. 555–569, Jun. 2022, DOI: 10.1007/s10551-021-04798-0.
\bibitem{3} L. Yu, H. Li, W. He, F.-K. Wang, and S. Jiao, ``A meta-analysis to explore privacy cognition and information disclosure of internet users,” International Journal of Information Management, vol. 51, p. 102015, Apr. 2020, DOI: 10.1016/j.ijinfomgt.2019.09.011.
\bibitem{3'}C. Perera, R. Ranjan, L. Wang, S. U. Khan, and A. Y. Zomaya, ``Big Data Privacy in the Internet of Things Era,” IT Professional, vol. 17, no. 3, pp. 32–39, May 2015, DOI: 10.1109/MITP.2015.34.
\bibitem{5'}A. A. Press, ``Google to pay \$60m fine for misleading Australians about collecting location data,” The Guardian, Aug. 12, 2022. Accessed: Sep. 24, 2022. [Online] Available: https://www.theguardian.com/technology/2022/aug/12/google-to-pay-60m-fine-for-misleading-australians-about-collecting-location-data.
\bibitem{6'} J. Golbeck and M.L. Mauriello, ``User Perception of Facebook App Data Access: A Comparison of Methods and Privacy Concerns", 2014, [Online] Available: http://hcil2.cs.umd.edu/trs/2014-19/2014-19.pdf.
\bibitem{ja} L. Na, C. Yang, C.-C. Lo, F. Zhao, Y. Fukuoka, and A. Aswani, ``Feasibility of reidentifying individuals in large national physical activity data sets from which protected health information has been removed with use of machine learning'', JAMA Network Open, vol. 1, no. 8, pp. e186040--e186040, 2018. 
\bibitem{4}“2022 Consumer Privacy Legislation.” https://www.ncsl.org/research/telecommunications-and-information-technology/2022-consumer-privacy-legislation.aspx.
\bibitem{7'}“Regulation (EU) 2016/679 of the European Parliament and of the Council of 27 April 2016.” http://cnpd.public.lu/en/legislation/droit-europ/union-europeenne/rgpd.html.
\bibitem{5}“Newsom wants companies collecting personal data to share the wealth with Californians,” Los Angeles Times, May 05, 2019. https://www.latimes.com/politics/la-pol-ca-gavin-newsom-california-data-dividend-20190505-story.html.
\bibitem{8'}H. Journal, “Oregon Health Information Property Act Proposes Paying Patients to Share Their Healthcare Data,” HIPAA Journal, Jan. 31, 2019. https://www.hipaajournal.com/oregon-health-information-property-act-proposes-paying-patients-to-share-their-healthcare-data/.
\bibitem{9'}H. Tsukayama, “Why Getting Paid for Your Data Is a Bad Deal,” Electronic Frontier Foundation, Oct. 26, 2020. https://www.eff.org/deeplinks/2020/10/why-getting-paid-your-data-bad-deal.
\bibitem{10'}R. Bartlett, A. Morse, R. Stanton, and N. Wallace, “Consumer-lending discrimination in the FinTech Era,” Journal of Financial Economics, vol. 143, no. 1, pp. 30–56, Jan. 2022, DOI: 10.1016/j.jfineco.2021.05.047.
\bibitem{6}E. Bax, “Computing a Data Dividend.” arXiv, Jun. 27, 2019. DOI: 10.48550/arXiv.1905.01805.
\bibitem{7}E. Bax et al., “Data Consortia,” in Proceedings of the Future Technologies Conference (FTC) 2020, Volume 2, Cham, 2021, pp. 489–498. DOI: $10.1007/978-3-030-63089-8\_31.$
\bibitem{11'}“A Data Dividend that Works: Steps Toward Building an Equitable Data Economy - Ideas - Berggruen Institute,” May 05, 2021. https://www.berggruen.org/ideas/articles/a-data-dividend-that-works-steps-toward-building-an-equitable-data-economy/, https://www.berggruen.org/ideas/articles/a-data-dividend-that-works-steps-toward-building-an-equitable-data-economy/.
\bibitem{12'}M. Simaan, “Equilibrium properties of the Nash and Stackelberg strategies,” Automatica, vol. 13, no. 6, pp. 635–636, Nov. 1977, DOI: 10.1016/0005-1098(77)90086-3.
\bibitem{13'}M. Simaan and J. B. Cruz, “On the Stackelberg strategy in nonzero-sum games,” J Optim Theory Appl, vol. 11, no. 5, pp. 533–555, May 1973, DOI: 10.1007/BF00935665.
\bibitem{14'}T. Basar, “On the relative leadership property of Stackelberg strategies,” J Optim Theory Appl, vol. 11, no. 6, pp. 655–661, 1973, DOI: 10.1007/BF00935564.
\bibitem{8} Y. Lee and A. Aswani, “Optimally Designing Cybersecurity Insurance Contracts to Encourage the Sharing of Medical Data,” in Proceedings of the Conference on Decision and Control, 2022. To appear.
\end{thebibliography}
\end{document}